\documentclass[aps,prl,twocolumn,superscriptaddress, showpacs,10pt]{revtex4}

\usepackage{times}
\usepackage{color,graphicx}
\usepackage{subfig,array}
\usepackage[hyperindex,breaklinks,colorlinks=true]{hyperref}
\usepackage{amsthm,amssymb,amsmath,amsfonts}
\usepackage{cleveref}
\usepackage{amscd}
\usepackage{lineno}

\newcommand\ket[1]{\ensuremath{|#1\rangle}}
\newcommand\bra[1]{\ensuremath{\langle#1|}}

\newcommand\tr{\mathop{\rm tr}\nolimits}

\def\a{\mathbf{a}}

\newcommand{\appsection}[1]{\let\oldthesection\thesection
  \renewcommand{\thesection}{Appendix \oldthesection}
  \section{#1}\let\thesection\oldthesection}

\newtheorem{lemma}{Lemma}
\newtheorem{theorem}{Theorem}

\def\Dbar{\leavevmode\lower.6ex\hbox to 0pt
{\hskip-.23ex\accent"16\hss}D}

\newcommand{\nc}{\newcommand}

\nc{\cA}{{\cal A}} \nc{\cB}{{\cal B}} \nc{\cC}{{\cal C}}
\nc{\cD}{{\cal D}} \nc{\cE}{{\cal E}} \nc{\cF}{{\cal F}}
\nc{\cG}{{\cal G}} \nc{\cH}{{\cal H}} \nc{\cI}{{\cal I}}
\nc{\cJ}{{\cal J}} \nc{\cK}{{\cal K}} \nc{\cL}{{\cal L}}
\nc{\cM}{{\cal M}} \nc{\cN}{{\cal N}} \nc{\cO}{{\cal O}}
\nc{\cP}{{\cal P}} \nc{\cQ}{{\cal Q}} \nc{\cR}{{\cal R}}
\nc{\cS}{{\cal S}} \nc{\cT}{{\cal T}} \nc{\cU}{{\cal U}}
\nc{\cV}{{\cal V}} \nc{\cW}{{\cal W}} \nc{\cX}{{\cal X}}
\nc{\cZ}{{\cal Z}}

\def\a{\alpha}
\def\b{\beta}

\def\c{\chi}

\begin{document}


\title{Symmetric Extension of Two-Qubit States}

\author{Jianxin Chen}%
\affiliation{Department of Mathematics \& Statistics, University of
  Guelph, Guelph, Ontario, Canada}%
\affiliation{Institute for Quantum Computing, University of Waterloo,
  Waterloo, Ontario, Canada}%
\author{Zhengfeng Ji}%
\affiliation{Institute for Quantum Computing, University of Waterloo,
  Waterloo, Ontario, Canada}%
\affiliation{State Key Laboratory of Computer Science, Institute of
  Software, Chinese Academy of Sciences, Beijing, China}
\author{David Kribs}%
\affiliation{Department of Mathematics \& Statistics, University of
  Guelph, Guelph, Ontario, Canada}%
\affiliation{Institute for Quantum Computing, University of Waterloo,
  Waterloo, Ontario, Canada}%
\author{Norbert L\"{u}tkenhaus}%
\affiliation{Department of Physics \& Astronomy, University of Waterloo, Waterloo, Ontario, Canada}
\affiliation{Institute for Quantum Computing, University of Waterloo,
  Waterloo, Ontario, Canada}%
\author{Bei Zeng}%
\affiliation{Department of Mathematics \& Statistics, University of
  Guelph, Guelph, Ontario, Canada}%
\affiliation{Institute for Quantum Computing, University of Waterloo,
  Waterloo, Ontario, Canada}%
\affiliation{Department of Physics \& Astronomy, University of Waterloo, Waterloo, Ontario, Canada}

\begin{abstract}
 A bipartite state $\rho_{AB}$ is symmetric extendible if there
  exists a tripartite state $\rho_{ABB'}$ whose $AB$ and $AB'$
  marginal states are both identical to $\rho_{AB}$. Symmetric
  extendibility of bipartite states is of vital importance in quantum
  information because of its central role in 
  separability tests, one-way distillation of EPR pairs, one-way
  distillation of secure keys, quantum marginal problems, and
  anti-degradable quantum channels. We establish a simple analytic
  characterization for symmetric extendibility of any two-qubit quantum state
  $\rho_{AB}$; specifically, $\tr(\rho_B^2) \geq \tr(\rho_{AB}^2) - 4
  \sqrt{\det{\rho_{AB}}}$. Given the intimate relationship between the
  symmetric extension problem and the quantum marginal problem, our
  result also provides the first analytic necessary and sufficient
  condition for the quantum marginal problem with overlapping
  marginals.
\end{abstract}

\date{\today}


\maketitle
The notion of {\em symmetric extendibility} for a bipartite quantum state $\rho_{AB}$ was introduced in~\cite{doherty02a} as a test for entanglement. A bipartite density operator $\rho_{AB}$ is symmetric extendible if there exists a tripartite state $\rho_{ABB'}$ such that $\tr_{B'}(\rho_{ABB'})=\tr_B(\rho_{ABB'})$. A state $\rho_{AB}$ without symmetric extension is evidently entangled, and to decide such an extendibility for $\rho_{AB}$ can be formulated in terms of semi-definite programming (SDP)~\cite{VB96}. This then leads to effective numerical tests and bounds~\cite{DPS04,DPS05,NOP09,BC12} that allow for entanglement detection for some well-known positive-partial-transpose (PPT) states~\cite{Per96,HHH96,HHH01,Hor97,HL00}. 

States with symmetric extension also have a clear operational meaning for quantum information processing~\cite{BDC03}. One simple idea is that if a bipartite state $\rho_{AB}$ is symmetric extendible, then one cannot distill any entanglement from $\rho_{AB}$ by protocols only involving local operations and one-way classical communication (from $A$ to $B$)~\cite{NL09}, because of entanglement monogamy~\cite{BDE+98}. Furthermore, using the Choi-Jamiolkowski isomorphism, symmetric extendibility of bipartite states also provides a test for anti-degradable quantum channels~\cite{Myhr11}, and one-way quantum capacity of quantum channels~\cite{NL09}. 

A similar idea applies to the protocols for quantum key distribution (QKD), which aim to establish a shared secret key between two parties (for a review, see~\cite{scarani09a}). The corresponding QKD protocols can be viewed as having two phases: in a first phase, the two parties establish joint classical correlations by performing measurements on an untrusted  bipartite quantum state, while in a second phase a secret key is being distilled from these correlations by a public discussion protocol (via authenticated classical channels) which typically involves classical error correction and privacy amplification~\cite{cachin97a, gottesman03a, chau02a,kraus05a}. If the underlying bipartite state $\rho_{AB}$ is symmetric extendible, then no secret key can be distilled by a process involving only one-way communication. Therefore, the foremost task of the public discussion protocol is to break this symmetric extendibility by some bi-directional post-selection process. Failure to find such a protocol means that no secret key can be established~\cite{MRDL09,ML09, Myhr11}.

From each of these perspectives then, we draw motivation for considering the symmetric extension problem, which asks for a characterization of all bipartite quantum states that possess symmetric extensions. Although the SDP formulation does provide an effective numerical tool, one always hopes for analytical results to provide a complete picture.

For the simplest case in which $\rho_{AB}$ is a two-qubit state, it is conjectured in~\cite{ML09} that the set $\rho_{AB}$ is symmetric extendible if and only if the spectra condition $\tr(\rho_B^2)\geq \tr(\rho_{AB}^2)-4\sqrt{\det{\rho_{AB}}}$ is satisfied. This elegant inequality is arrived at by studying several examples, both analytically and numerically; for example the Bell diagonal states, and the $ZZ$-invariant states. Unfortunately, \cite{ML09} fails to prove in general either the necessity or the sufficiency of the conjecture, an unusual situation as typically one of the directions would be easy to establish. This hints at an intrinsic hardness to the problem, whose solution may require new physical insight.

It has been observed that the symmetric extension problem is a special case of the quantum marginal problem~\cite{Myhr11}, which asks for the conditions under which some set of density matrices $\{\rho_{A_i}\}$ for the subsets $A_i\subset\{1,2,\ldots,n\}$ are reduced density matrices of some state $\rho$ of the
whole $n$-particle system~\cite{Kly06}. The related problem in fermionic (bosonic) systems is the so-called $N$-representability problem, which inherits a long history in quantum chemistry~\cite{Col63,Erd72}.

Succinct necessary and sufficient conditions are obtained for the quantum marginal problem and the
$N$-representability problem for non-overlapping marginals~\cite{Kly06,AK08}. However, the overlapping version, which includes our symmetric extension problem as a special case, turns out to be much more difficult~\cite{Col63}. It was shown that the overlapping marginal problem belongs to the complexity class of QMA-complete, even for the relatively simple case where the marginals $\{\rho_{A_i}\}$ are two-particle density matrices~\cite{Liu06,LCV07,WMN10}. Nevertheless, the solution to small systems would provide insight on developing approximation/numerical methods for larger systems, though on the analytical side only a handful partial results are known~\cite{Smi65,CLL13}.

In this work, we prove the conjecture that a two-qubit state $\rho_{AB}$ is symmetric extendible
if and only if $\tr(\rho_B^2)\geq \tr(\rho_{AB}^2)-4\sqrt{\det{\rho_{AB}}}$. Our main insight for obtaining this result relies largely on the physical pictures from the study of the quantum marginal problem. Besides providing a better understanding for various quantum information protocols related to symmetric extension, our result also gives the first analytic necessary and sufficient condition for the quantum marginal problem with overlapping marginals.


\textit{Symmetric extension}-- For
any two-qubit state $\rho_{AB}$, denote its symmetric extension
by $\rho_{ABB'}$ (may be non-unique), hence $\rho_{AB}=\rho_{AB'}$. Consider the
following set
\begin{equation}
\label{eq:pure}
\mathcal{A}=\{\rho_{AB}: \vec{\lambda}(\rho_{AB})=\vec{\lambda}(\rho_B)\},
\end{equation}
where $\vec{\lambda}(\rho)$ denotes the nonzero eigenvalues of $\rho$
in decreasing order. It is shown in~\cite{ML09} that
$\mathcal{A}$ fully characterizes the set of two-qubit states
which admit pure symmetric extension
$\rho_{ABB'}=\ket{\psi_{ABB'}}\bra{\psi_{ABB'}}$
for some pure state $\ket{\psi_{ABB'}}$.
This follows from the Schmidt
decomposition of $\ket{\psi_{ABB'}}$, which gives the same nonzero spectra
for $\rho_{AB}$ and $\rho_{B}$.

The convex hull of $\mathcal{A}$ is given by
\begin{eqnarray*}
\label{eq:hull}
\mathcal{B}&=&\{\rho_{AB}: \rho_{AB}=\sum\limits_{j}p_j\rho_{AB}^j;\nonumber\\
&& 0\leq p_j\leq 1; \sum\limits_j p_j=1; \rho_{AB}^j\in \mathcal{A}\},
\end{eqnarray*}
which completely characterizes the set of two-qubit states
that admit symmetric extension.

It is conjectured in~\cite{ML09} that the set $\mathcal{B}$ may be equal
to another analytically tractable set $\mathcal{C}$ given by
\begin{eqnarray}
\label{eq:conj}
\mathcal{C}=\big\{\rho_{AB}: \tr(\rho_B^2)\geq \tr(\rho_{AB}^2)-4\sqrt{\det{\rho_{AB}}}\big\}.
\end{eqnarray}

Our main result is to show that the conjecture $\mathcal{B}=\mathcal{C}$ is indeed valid; that is,
\begin{theorem}
\label{th:main}
A two qubit state $\rho_{AB}$ admits a symmetric extension if and only if $\tr(\rho_B^2)\geq \tr(\rho_{AB}^2)-4\sqrt{\det{\rho_{AB}}}$.
\end{theorem}

Our key insight for obtaining this result relies largely on the physical pictures from the study of the quantum marginal problem, regarding the structure of $\mathcal{B}$. Notice that $\mathcal{B}$ is a convex set, therefore for any point $\sigma_{AB}\in\partial\mathcal{B}$, where $\partial{\mathcal{B}}$ denotes the boundary of ${\mathcal{B}}$, there exists a supporting hyperplane through $\sigma_{AB}$, which is associated with an observable $H_{AB}(\sigma_{AB})$. That is, $\tr\big(H_{AB}(\sigma_{AB})\cdot\rho_{AB}\big)\geq 0$ holds for any $\rho_{AB}\in \mathcal{B}$. This induces a Hamiltonian $H=H_{AB}+H_{AB'}$ for the three-qubit system $ABB'$, which has the symmetric extension $\rho_{ABB'}$ supported on the ground-state space of $H$.

If it were indeed the case that $\mathcal{B}=\mathcal{C}$, then $\mathcal{C}$ must inherit all the above-mentioned properties of the convex body $\mathcal{B}$. These observations then hint for the structure of the intersection of $\partial\mathcal{C}$ with the supporting hyperplane associated with $H_{AB}(\sigma_{AB})$, which are in fact faces of the convex body $\mathcal{C}$.

\textit{The necessary condition}-- We first prove the necessary condition of Theorem~\ref{th:main}, which, we will observe below, will follow if we prove $\mathcal{C}$ is convex. A natural approach here would be to
assume that for any $\rho_{AB},\sigma_{AB}\in\mathcal{C}$, their convex combination
$p\rho_{AB}+(1-p)\sigma_{AB}$ for any $p\in[0,1]$ is also in $\mathcal{C}$. However the characterization
of $\mathcal{C}$ by Eq.~\eqref{eq:conj} involves the square root of a determinant, which is
not easy to handle directly.

We instead take another slightly different approach.
Our idea is to use the fact that a closed set with nonempty interior is convex if every point on its boundary has a supporting hyperplane~\cite{boyd2004convex}. Therefore our goal is to find such a supporting hyperplane for any $\sigma_{AB}\in\partial{\mathcal{C}}$.

To achieve our goal, we will need to characterize the boundary of $\mathcal{C}$ (i.e. $\partial{\mathcal{C}}$). Let $f(\sigma_{AB})=\tr(\sigma_B^2)- \tr(\sigma_{AB}^2)+4\sqrt{\det{\sigma_{AB}}}$. We have the following result.  
\begin{lemma}
\label{lm:pC}
$\partial{\mathcal{C}}$ contains all states $\sigma_{AB}\in\mathcal{C}$ without full rank  (i.e. has rank $<4$) and all full rank states $\sigma_{AB}\in\mathcal{C}$ satisfying $f(\sigma_{AB})=0$.
\end{lemma}

To show the validity of Lemma~\ref{lm:pC}, we first consider the case where $\sigma_{AB}$ is without full rank.  Consider the polynomial
$
\det(y\rho_{AB}+\sigma_{AB})=\sum_{k=0}^4 c_k(\rho_{AB})y^k
$
for $\rho_{AB}\in{\mathcal{C}}$.
Define $h(\rho_{AB})=c_1(\rho_{AB})$. Notice that $c_0(\rho_{AB})=0$,  and
$\det(y\rho_{AB}+\sigma_{AB})\geq 0$ when $y\rightarrow 0^{+}$. Furthermore, $h(\sigma_{AB})=0$. This implies that $\{X: h(X)=0\}$ is a supporting hyperplane at $\sigma_{AB}$. Hence it follows that any $\sigma_{AB}\in\mathcal{C}$ without full rank is in $\partial{\mathcal{C}}$, and furthermore there is always a supporting hyperplane at $\sigma_{AB}$.


We then discuss the case that $\sigma_{AB}\in \partial{\mathcal{C}}$
is of full rank (i.e. rank $4$).  In this case, we show that all $\sigma_{AB}\in \partial{\mathcal{C}}$ are characterized by $f(\sigma_{AB})=0$. To see this, notice that $\sigma_{AB}$ lies on the boundary if every neighbourhood of $\sigma_{AB}$ contains at least one point in $\mathcal{C}$ and at least one point not in $\mathcal{C}$.

For any Hermitian operator $M_{AB}$, we have the following expansion by Jacobi's formula (see e.g.~\cite{magnus1988matrix}),
\begin{eqnarray}
\label{eq:Jacobi}
&&f(\sigma_{AB}+\epsilon M_{AB})-f(\sigma_{AB})\nonumber\\
&=&2 \tr\big(H_{AB}(\sigma_{AB})\cdot M_{AB}\big)\epsilon+O(\epsilon^2),
\end{eqnarray}
where
\begin{equation}
\label{eq:H}
H_{AB}(\sigma_{AB})=\sqrt{\det\sigma_{AB}} \sigma_{AB}^{-1}- \sigma_{AB}+\sigma_B.
\end{equation}

Now for any full rank state $\sigma_{AB}$ satisfying the strict inequality $f(\sigma_{AB})>0$, we can always find an open ball centered at $\sigma_{AB}$ over which the strict inequality always holds; i.e., $\sigma_{AB}$ is an interior point. On the other hand, if $f(\sigma_{AB})=0$, then we can always choose suitable Hermitian operators $M_{AB}, M_{AB}^{\prime}$ such that $\tr\big(H_{AB}(\sigma_{AB}) \cdot M_{AB}\big)>0$ and $\tr\big(H_{AB}(\sigma_{AB}) \cdot M_{AB}^{\prime}\big)<0$ unless $H_{AB}(\sigma_{AB})=0$. The latter cannot occur, as it would imply $\sqrt{\det{\sigma_{AB}}} \mathbb{I}_{AB}=\sigma_{AB}^2-\sigma_{AB}^{\frac{1}{2}}\sigma_B\sigma_{AB}^{\frac{1}{2}}\leq\sigma_{AB}^2$. The last inequality holds only if $\sigma_{AB} \propto \mathbb{I}_{AB}$ which immediately contradicts $f(\frac{\mathbb{I}}{4})=\frac{1}{2}$. Hence any full rank states satisfying $f(\sigma_{AB})=0$ are boundary points of $\mathcal{C}$.

Provided the full characterization of $\partial{\mathcal{C}}$ given by Lemma~\ref{lm:pC}, especially the form of Eqs.~\eqref{eq:Jacobi} and ~\eqref{eq:H}, our main result of this section is then the following theorem.
\begin{theorem}\label{thm:hyperplane}
For any full rank state $\sigma_{AB}\in\partial{\mathcal{C}}$ and $H_{AB}(\sigma_{AB})$ as given in Eq.~\eqref{eq:H}, the inequality
\begin{equation}
\label{eq:hyperplane}
\tr\big(H_{AB}(\sigma_{AB})\cdot \rho_{AB}\big)\geq 0
\end{equation}
holds for any $\rho_{AB}\in \mathcal{C}$.
\end{theorem}

Note that the equality of Eq.~\eqref{eq:hyperplane} holds when $\rho_{AB}=\sigma_{AB}$.
Eq.~\eqref{eq:hyperplane} then means that
for any full rank $\sigma_{AB}\in \partial{\mathcal{C}}$,
there is a supporting hyperplane of $\mathcal{C}$ which can be characterized by
\begin{eqnarray}
\label{eq:L}
\mathcal{L}(\sigma_{AB}):=
\big\{X: \tr \big( H_{AB}(\sigma_{AB}) \cdot X\big)=0\big\}.
\end{eqnarray}
In order to show the validity of Eq.~\eqref{eq:hyperplane},
we will need another characterization of the set $\mathcal{C}$,
and follow a straightforward step-by-step optimization procedure that involves a lengthy calculation.
We provide the technical details in Appendix~I \& II.

To summarize, we have thus shown that for any $\sigma_{AB}\in\partial{\mathcal{C}}$, with or without full rank, there exists a supporting hyperplane at $\sigma_{AB}$. This then concludes the proof that $\mathcal{C}$ is convex.

A direct consequence of the convexity of $\mathcal{C}$ is that $\mathcal{B}\subseteq \mathcal{C}$. To see why, we can easily verify that $\mathcal{A}\subset \mathcal{C}$. Additionally, $\mathcal{B}$ is the convex hull of $\mathcal{A}$. Therefore the convex hull of $\mathcal{A}$ is a subset of the convex hull of $\mathcal{C}$ which is again $\mathcal{C}$, and thus we have $\mathcal{B}\subseteq\mathcal{C}$, as required.

\textit{The sufficient condition}--
To prove the sufficiency of Theorem~\ref{th:main},
we will need to show that any state in $\mathcal{C}$ can be
represented as a convex combination of some states in $\mathcal{A}$.
In fact, given the convexity of $\mathcal{C}$, it suffices to show
this for $\sigma_{AB}\in\partial{\mathcal{C}}$.

Furthermore, we only need to deal with the cases where $\sigma_{AB}\in\partial{\mathcal{C}}$ is of full rank or rank $3$. The rank $1$ case is obvious and the rank $2$ case has already been solved in~\cite{ML09}. That is, any rank $2$ state $\rho_{AB}\in{\mathcal{C}}$ can be written as a convex combination of two states in ${\mathcal{A}}$, hence $\rho_{AB}$ is symmetric extendible.

For the full rank case, let us first build up some intuition by imagining what should happen if $\mathcal{B}=\mathcal{C}$. According to Eq.~\eqref{eq:hyperplane}, for any $\sigma_{AB}\in\partial{\mathcal{C}}$,
there exists a supporting hyperplane $\mathcal{L}(\sigma_{AB})$ given by all the $X$ satisfying
$\tr \big( H_{AB}(\sigma_{AB}) \cdot X\big)=0$, where $H_{AB}(\sigma_{AB})$ given in Eq.~\eqref{eq:H} is a Hermitian operator acting on the qubits $A$ and $B$.

Now let us consider the following operator $H$ acting on the three-qubit system $ABB'$:
\begin{equation}
H=H_{AB}+H_{AB'}.
\end{equation}

Note that $H$ can be viewed as a Hamiltonian of the system $ABB'$. The symmetric extension of $\sigma_{AB}$, denoted by $\sigma_{ABB'}$, should have energy zero as $\tr(H\sigma_{ABB'})=\tr(H_{AB}\sigma_{AB})+\tr(H_{AB'}\sigma_{AB'})$.

Furthermore, we show that $H$ is positive. Since $H$ is symmetric when swapping $BB'$, we can always find a complete set of eigenstates $\{\ket{\psi_i}\}_{i=1}^8$ of $H$, such that for each $\ket{\psi_i}$,
$\tr_B(\ket{\psi_i}\bra{\psi_i})=\tr_{B'}(\ket{\psi_i}\bra{\psi_i})$.
This is because if there is any eigenstate $\ket{\phi}$ of $H$ with energy $E_{\phi}$ which does not satisfy $\tr_B(\ket{\phi}\bra{\phi})=\tr_{B'}(\ket{\phi}\bra{\phi})$, then the state $\ket{\phi'}=S\ket{\phi}$ is also an eigenstate of $H$ with the same energy $E_{\phi}$, where $S:=\text{SWAP}_{BB'}$ is the swap operation acting on the qubits $BB'$. Therefore we can re-choose the eigenstates with energy $E_{\phi}$ as $\varphi=1/\sqrt{2}(\ket{\phi}+\ket{\phi'})$ and
$\varphi'=1/\sqrt{2}(\ket{\phi}-\ket{\phi'})$,
then we will have $\tr_B(\ket{\varphi}\bra{\varphi})=\tr_{B'}(\ket{\varphi}\bra{\varphi})$ and $\tr_B(\ket{\varphi'}\bra{\varphi'})=\tr_{B'}(\ket{\varphi'}\bra{\varphi'})$.

It then directly follows from Eq.~\eqref{eq:hyperplane} that for this complete set of eigenstates $\{\ket{\psi_i}\}_{i=1}^8$ with $\tr_B(\ket{\psi_i}\bra{\psi_i})=\tr_{B'}(\ket{\psi_i}\bra{\psi_i})$, $\tr(H\ket{\psi_i}\bra{\psi_i})\geq 0$. That is, $H$ is positive. Therefore, $\sigma_{ABB'}$ with zero energy is supported on the ground-state space of $H$ (for general discussion on supporting hyperplanes and the ground-state space, see e.g.~\cite{Erd72,CJR+12,CJZZ12}).


Because $H$ is symmetric when swapping $BB'$, generically, the ground-state space of $H$ should be doubly degenerate. To see this, if $\ket{\psi_0}$ is a ground state of $H$, $S\ket{\psi_0}$ is also a ground state of $H$. And generically, $S\ket{\psi_0}$ should be linear independent from $\ket{\psi_0}$.

Let us now denote the ground-state space of $H$ by $V_H$, which is generically two-dimensional, and define
\begin{equation*}
\mathcal{F}:=\big\{\rho_{AB}|\rho_{AB}=\tr_{B'}\rho_{ABB'},\rho_{ABB'}\ \text{supported on}\ V_H\big\}.
\end{equation*}

Note that $\mathcal{F}\subset \partial\mathcal{C}$ and $\mathcal{F}$ is in fact a face of the convex body $\mathcal{C}$. And we have that for $\sigma_{AB}\in\mathcal{F}$,
its symmetric extension $\sigma_{ABB'}$ is supported on the ground-state space of $H$.
This indicates that $\mathcal{F}=\mathcal{L}(\sigma_{AB})\bigcap \partial{\mathcal{C}}$.

Because $V_H$ is generically two-dimensional, any state supported on $V_H$ can be parameterized by a two-dimensional unitary operator $U$, and the two eigenvalues $\lambda_0,\lambda_1$
of any state that is supported on $V_H$ (with $\lambda_0+\lambda_1=1$). That is, any state $\rho_{ABB'}$ supported on $V_H$, is of the form, in some chosen orthonormal basis of $\{\ket{\psi_1},\ket{\psi_2}\}$ of $V_H$, as
\begin{equation*}
\rho_{ABB'}(\lambda_0,\lambda_1,U)=U(\lambda_0\ket{\psi_0}\bra{\psi_0}+\lambda_1\ket{\psi_1}\bra{\psi_1})U^{\dag}.
\end{equation*}
Consequently, any state $\rho_{AB}=\tr_{B'}\rho_{ABB'}\in\mathcal{L}(\sigma_{AB})\bigcap \partial{\mathcal{C}}$ can be also parametrized by $\lambda_0,\lambda_1,U$, that
we can denote as $\rho_{AB}(\lambda_0,\lambda_1,U)$.

Furthermore, any $\rho_{ABB'}(\lambda_0,\lambda_1,U)$ has the obvious decomposition
\begin{eqnarray*}
\rho_{ABB'}(\lambda_0,\lambda_1,U)
=\lambda_0\rho_{ABB'}(1,0,U)+\lambda_1\rho_{ABB'}(0,1,U),
\end{eqnarray*}
where both $\rho_{ABB'}(1,0,U)$ and $\rho_{ABB'}(0,1,U)$
are three-qubit pure states. As a result,
\begin{equation*}
\rho_{AB}(\lambda_0,\lambda_1,U)=\lambda_0\rho_{AB}(1,0,U)+\lambda_1\rho_{AB}(0,1,U),
\end{equation*}
where both $\rho_{AB}(1,0,U)$ and $\rho_{AB'}(0,1,U)$ are in $\partial\mathcal{C}$
and of rank $2$.

Summarizing the discussion above, for a full rank $\sigma_{AB}\in\partial\mathcal{C}$, we shall expect that generically, any state in $\mathcal{L}(\sigma_{AB})\bigcap \partial{\mathcal{C}}$ can be parameterized by a two-dimensional unitary $U$ and two real parameters $\lambda_0,\lambda_1$, denoted as $\rho_{AB}(\lambda_0,\lambda_1,U)$. And any such $\rho_{AB}(\lambda_0,\lambda_1,U)$ can always  be written as a convex combination of two rank $2$ states in $\partial\mathcal{C}$. Detailed analysis of $\mathcal{L}(\sigma_{AB})\bigcap \partial{\mathcal{C}}$ shows this is not only generically the case, but also always the case. This is given as the following theorem. We shall provide the technical details of the proof in Appendix~III). 
\begin{theorem}
\label{thm:suf}
Every full rank $\sigma_{AB}\in\partial\mathcal{C}$ can be written as a convex combination of two rank $2$ states in $\partial\mathcal{C}$.
\end{theorem}

Furthermore, because any rank $2$ state $\rho_{AB}\in{\mathcal{C}}$ can be written as a convex combination of two states in ${\mathcal{A}}$, it follows that any full rank $\sigma_{AB}\in\partial\mathcal{C}$ can be written as a convex combination of states in ${\mathcal{A}}$, hence is symmetric extendible.

Now consider the case where $\sigma_{AB}\in \partial{\mathcal{C}}$ has rank $3$.
Let $\ket{\phi}$ be the state in $\ker{\sigma_{AB}}$. Notice that since any two-qubit state is local unitary equivalent to the state $a\ket{00}+b\ket{11}$ for some $a,b$, we can always write $\sigma_{AB}$ in the following form without loss of generality:
\begin{eqnarray*}
\sigma_{AB}=\left(\begin{array}{cccc}
\vert b\vert^2& b^{\ast}x^{\ast}& b^{\ast}y^{\ast} & -a b^{\ast} \\
bx & \vert r\vert^2 &t & -ax\\
by& t^{\ast}& \vert s\vert^2 & -ay\\
-a^{\ast}b& -a^{\ast}x^{\ast}& -a^{\ast}y^{\ast}& \vert a\vert^2
\end{array}\right).
\end{eqnarray*}

Let us choose the Hermitian operator
\begin{equation*}
M_{AB}=\left(\begin{array}{cccc}
0 & b^{\ast}p^{\ast} & b^{\ast}q^{\ast} & 0\\
bp & 0 & 0 & -ap\\
bq & 0 & 0 & -aq\\
0 & -a^{\ast}p^{\ast} & -a^{\ast}q^{\ast} & 0
\end{array}\right),
\end{equation*}
where $p,q$ are constants to be fixed later and define $\sigma(\epsilon)=\sigma_{AB}+\epsilon M_{AB}$.


Then
\begin{eqnarray}
&&\tr (\sigma(\epsilon)_B^2)-\tr (\sigma(\epsilon)_{AB}^2)\nonumber\\
&=&\tr((\sigma_B+\epsilon M_B)^2)- \tr((\sigma_{AB}+\epsilon M_{AB})^2)\nonumber\\
&=&\tr(\sigma_B^2)-\tr(\sigma_{AB}^2)+\epsilon^2(\tr(M_B^2)-\tr(M_{AB}^2))\nonumber\\
&&+2\epsilon (\tr(\sigma_BM_B)-\tr(\sigma_{AB}M_{AB}))\nonumber\\
&=&-2\vert ap+b^{\ast}q^{\ast}\vert^2 \epsilon^2-4\Re\big((ap+b^{\ast}q^{\ast})(a^{\ast}x^{\ast}+by)\big)\epsilon , \nonumber
\end{eqnarray}
where $\Re$ stands for the real part of a complex number.

By choosing suitable $p, q$ such that $ap+b^{\ast}q^{\ast}=0$, we will  have $\tr (\sigma(\epsilon)_B^2)=\tr(\sigma(\epsilon)_{AB}^2)$ which implies $\sigma(\epsilon) \in \partial{\mathcal{C}}$ if $\sigma(\epsilon)$ is a density operator. $M_{AB}$ is a traceless operator whose kernel also contains $\ket{\phi}$, therefore with growing $\epsilon$ in either direction, we will have positive $\epsilon_{+}$ and negative $\epsilon_{-}$ such that $\sigma(\epsilon_i)=\sigma_{AB}+\epsilon_i M_{AB} \in \partial{\mathcal{C}}$ and $\text{rank}(\sigma(\epsilon_i))\leq 2$ for any $i\in\{+,-\}$. Hence, $\sigma_{AB}$ of rank $3$ can be written as a convex combination of at most two states from $\mathcal{A}$.

This concludes the proof of the sufficiency condition of Theorem~\ref{th:main}.

\textit{Example}-- To better understand the physical picture, let us look at an example. Consider the two-qubit Werner state
\begin{equation*}
\rho_W(p)=(1-p)\frac{\mathbb{I}}{4}+p\ket{\phi}\bra{\phi},
\end{equation*}
where $\ket{\phi}=\frac{1}{\sqrt{2}}(\ket{00}+\ket{11})$, and $p\in[0,1]$. The equation 
\[
\tr\big(\rho^2_W(p)\big)= \tr\Big(\big(\tr_B\rho_W(p)\big)^2\Big)+4\sqrt{\det{\rho_W(p)}}
\] 
provides a unique solution of $p=\frac{2}{3}$; i.e., $\rho_W(\frac{2}{3})\in\partial\mathcal{C}$.
Further, Eq.~\eqref{eq:H} gives
\begin{equation*}
H_{AB}\left(\rho_W\left(\frac{2}{3}\right)\right)=
\begin{pmatrix}
\frac{2}{9}&0&0&-\frac{4}{9}\\
0&\frac{2}{3}&0&0\\
0&0&\frac{2}{3}&0\\
-\frac{4}{9}&0&0&\frac{2}{9}
\end{pmatrix}
\end{equation*}

The ground-state space of the Hamiltonian $H\left(\rho_W(\frac{2}{3})\right)=H_{AB}\left(\rho_W(\frac{2}{3})\right)+H_{AB'}\left(\rho_W(\frac{2}{3})\right)$ is indeed two-fold degenerate and
is spanned by
\begin{eqnarray*}
\ket{\psi_0}&=&\frac{1}{\sqrt{6}}\left(2\ket{000}+\ket{101}+\ket{110}\right),\nonumber\\
\ket{\psi_1}&=&\frac{1}{\sqrt{6}}\left(2\ket{111}+\ket{010}+\ket{001}\right).
\end{eqnarray*}

Therefore any state $\rho_{ABB'}$ supported on this ground-state space can be written as
$
\rho_{ABB'}(\lambda_0,\lambda_1,U)=U(\lambda_0\ket{\psi_0}\bra{\psi_0}+\lambda_1\ket{\psi_1}\bra{\psi_1})U^{\dag},
$
for some $2\times 2$ unitary operator $U$ acting on the two-dimensional space spanned by
$\ket{\psi_0},\ket{\psi_1}$.
And any state $\rho_{AB}$ in $\mathcal{L}\left(\rho_W(\frac{2}{3})\right)\bigcap \partial{\mathcal{C}}$ has the form
$
\rho_{AB}(\lambda_0,\lambda_1,U)
=\tr_{B'}\left(U(\lambda_0\ket{\psi_0}\bra{\psi_0}+\lambda_1\ket{\psi_1}\bra{\psi_1})U^{\dag}\right).
$

It is straightforward to check that $\rho_W(\frac{2}{3})=\rho_{AB}(\frac{1}{2},\frac{1}{2},\mathbb{I})$. In other words, the symmetric extension of $\rho_W(\frac{2}{3})$,
given by $\rho_{ABB'}(\frac{1}{2},\frac{1}{2},\mathbb{I})$, is the maximally mixed state of the ground-state space of $H\left(\rho_W(\frac{2}{3})\right)$.
Also, $\rho_W(\frac{2}{3})$ can clearly be written as the convex combination of two rank $2$ states which are also in $\mathcal{L}\left(\rho_W(\frac{2}{3})\right)\bigcap \partial{\mathcal{C}}$:
$
\rho_W\left(\frac{2}{3}\right)=\frac{1}{2}\rho_{AB}(1,0,\mathbb{I})+\frac{1}{2}\rho_{AB}(0,1,\mathbb{I}).
$

We remark that $p=\frac{2}{3}$ corresponds to a fidelity $\frac{3}{4}$ with $\ket{\phi}$. This is consistent with the result that Werner states with fidelity $\leq \frac{3}{4}$ have zero one-way distillable entanglement~\cite{BDSW96,KL97}.


\textit{Discussion} -- We have fully solved the symmetric extension problem for the two-qubit case,  
providing the first analytical necessary and sufficient condition for the quantum marginal problem with overlapping marginals.

An immediate application of our result is a full characterization for anti-degradable qubit channels, as it is known that a channel $\cal{N}$ is anti-degradable if and only if its Choi-Jamiolkowski representation $\rho_{\mathcal{N}}$ has a symmetric extension~\cite{Myhr11}. Previously analytic necessary and sufficient conditions were only known for anti-degradable unital qubit channels~\cite{CRS08,Cerf00,NG98}.

A natural question to ask is how to generalize the result to higher dimensional systems. Unfortunately, for any higher dimensions a full characterization involving only spectra is highly unlikely~\cite{ML09}.  There have been some efforts made for special cases but no general results found~\cite{Ran09a,Ran09b,JV13}. Nevertheless, our physical picture based on the convexity of $\mathcal{B}$ and the symmetry of the system may shed light on the understanding of symmetric extendibility for higher dimensional systems.


\textit{Acknowledgements}--
   The authors thank John Watrous and Debbie Leung for helpful
  discussions. JC and BZ are supported by NSERC, CIFAR and NSF of
  China (Grant No.61179030). ZJ acknowledges support from NSERC and ARO. DK was supported by NSERC and a University Research Chair at Guelph. NL acknowledges support from NSERC and ORF.
\bibliographystyle{apsrev4-1}
\bibliography{SymExt}



\appendix
\label{appendix}

\section{\label{appendix:char}Appendix~I: A Useful Characterization of $\mathcal{C}$}

To prove our main result, we will provide another useful characterization of $\mathcal{C}=\{\rho_{AB}: \tr(\rho_B^2)\geq
\tr(\rho_{AB}^2)-4\sqrt{\det{\rho_{AB}}}\}$, the set we are mainly interested in.

For simplicity, we use $\mathbb{M}_2$ to denote the set of $2$-by-$2$ matrices.

\begin{lemma}\label{lemma:char}
\begin{eqnarray*}
\mathcal{C}&=&\Bigg\{\left(
\begin{array}{cc}
Q  &  R   \\
P  &  0
\end{array}
\right)\left(
\begin{array}{cc}
Q^{\dagger}  &  P   \\
R  &  0
\end{array}
\right):P,Q,R\in\mathbb{M}_2 \mathrm{\ such\ that\ }   \\
&& P, R\geq 0 \mathrm{\ and\ }\Vert PR\Vert_{\mathrm{tr}}^2 \geq \Vert PQ^{\dagger}\Vert_{\mathrm{tr}}^2-\Vert PQ\Vert_{\mathrm{tr}}^2.\Bigg\}.
\end{eqnarray*}
\end{lemma}

\begin{proof}
Any mixed state $\rho_{AB}$ satisfying $\tr(\rho_B^2)\geq
\tr(\rho_{AB}^2)-4\sqrt{\det{\rho_{AB}}}$ can be written in the matrix form
$\left(
\begin{array}{cc}
A  & C   \\
C^{\dagger}  & B   \\
\end{array}
\right)$
where $B$ and $A$ are $2$ by $2$ positive semidefinite matrices and $C$ is another $2$ by $2$ matrix. We first assume $B$ is invertible, then $A$ can be written as $CB^{-1}C^{\dagger}+D$ where $D$ is another $2$ by $2$ positive semidefinite matrix.

Employing the following identity
\begin{eqnarray*}
\left(
\begin{array}{cc}
CB^{-1}C^{\dagger}+D  & C   \\
C^{\dagger}  & B   \\
\end{array}
\right)=\left(
\begin{array}{cc}
\mathbb{I}  & CB^{-1}   \\
0  & \mathbb{I}   \\
\end{array}
\right)\left(
\begin{array}{cc}
D  & 0   \\
C^{\dagger}  & B   \\
\end{array}
\right)
\end{eqnarray*}
leads to $\det{\rho_{AB}}=\det{(BD)}$

It is not hard to verify that $\tr(\rho_B^2)\geq
\tr(\rho_{AB}^2)-4\sqrt{\det{\rho_{AB}}}$ is equivalent to the condition that $\tr(BD)+2\sqrt{\det BD}\geq \tr (CC^{\dagger})-\tr (CB^{-1}C^{\dagger}B)$.

Observe that $\tr(BD)+2\sqrt{\det BD}=(\tr\sqrt{B^{\frac{1}{2}}DB^{\frac{1}{2}}})^2$, we can further let $D=B^{-\frac{1}{2}}X^2B^{-\frac{1}{2}}$ where $X$ is a positive semidefinite matrix.

Then
\begin{eqnarray}
\rho_{AB}=\left(
\begin{array}{cc}
CB^{-1}C^{\dagger}+B^{-\frac{1}{2}}X^{2}B^{-\frac{1}{2}}  & C   \\
C^{\dagger}  & B   \\
\end{array}
\right)
\end{eqnarray}
where $B$ and $X$ are $2$ by $2$  positive semidefinite matrices and $C$ is a $2$ by $2$ matrix and they satisfy
$(\tr X)^2\geq \tr (CC^{\dagger})-\tr (CB^{-1}C^{\dagger}B)$.

Let us write $C=B^{-\frac{1}{2}}YB^{\frac{1}{2}}$, we have
\begin{eqnarray*}
\rho_{AB}&=&\left(
\begin{array}{cc}
B^{-\frac{1}{2}}(YY^{\dagger}+X^{2})B^{-\frac{1}{2}}  & B^{-\frac{1}{2}}YB^{\frac{1}{2}}   \\
B^{\frac{1}{2}}Y^{\dagger}B^{-\frac{1}{2}}  & B   \\
\end{array}
\right)\\
&=&\left(
\begin{array}{cc}
B^{-\frac{1}{2}}   & 0  \\
0   & B^{\frac{1}{2}}
\end{array}
\right)
\left(
\begin{array}{cc}
Y   & X  \\
\mathbb{I}   & 0
\end{array}
\right)
\left(
\begin{array}{cc}
Y^{\dagger}& \mathbb{I}\\
X & 0
\end{array}
\right)
\left(
\begin{array}{cc}
B^{-\frac{1}{2}}   & 0  \\
0   & B^{\frac{1}{2}}
\end{array}
\right)\\
&=&\left(
\begin{array}{ccc}
B^{-\frac{1}{2}}Y  & B^{-\frac{1}{2}}X   \\
B^{\frac{1}{2}}  & 0   \\
\end{array}
\right){\left(
\begin{array}{ccc}
B^{-\frac{1}{2}}Y  & B^{-\frac{1}{2}}X   \\
B^{\frac{1}{2}}  & 0   \\
\end{array}
\right)}^{\dagger}
\end{eqnarray*}
where $X$ and $B$ are $2$ by $2$ positive semidefinite matrices and $Y$ is a $2$ by $2$ matrix and they satisfy
$(\tr{X})^2\geq \tr (B^{-1}YBY^{\dagger})-\tr (YY^{\dagger})$.

Therefore, any $\rho_{AB}\in \mathcal{C}$ can be written as
\begin{eqnarray}
\rho_{AB}=\left(
\begin{array}{cc}
Q  &  R   \\
P  &  0
\end{array}
\right)\left(
\begin{array}{cc}
Q^{\dagger}  &  P   \\
R^{\dagger}  &  0
\end{array}
\right)
\end{eqnarray}
where $Q$ and $R$ are $2$ by $2$ matrices and $Q$ is  $2$ by $2$ positive semidefinite matrix and they satisfy
$\Vert PR\Vert_{tr}^2=(\tr\sqrt{PRR^{\dagger}P})^2\geq \tr (PP(Q^{\dagger}Q-QQ^{\dagger}))$.

Furthermore, we can even choose $R$ to be a positive semidefinite matrix since $R$ only appears in the term $RR^{\dagger}$ of the top-left $2$-by-$2$ submatrix of $\rho_{AB}$.

Now let's look at the case that $B$ is singular. $B$ is thus a rank $1$ positive operator, without loss of generality, let's assume it is a rank $1$ projection $\ket{u}\bra{u}$. Follows from the positivity of $\rho_{AB}$, $C$ can be written as $\ket{u}\bra{v}$. Hence,
\begin{eqnarray*}
\rho_{AB}=\left(
\begin{array}{cc}
D+\ket{v}\bra{v} & \ket{v}\bra{u}   \\
\ket{u}\bra{v}  & \ket{u}\bra{u}   \\
\end{array}
\right)
\end{eqnarray*}
where $\ket{u}$ is a unit vector, but $\ket{v}$ is unnormalized.

We can simply choose $P=\ket{u}\bra{u}$, $Q=\ket{v}\bra{u}$ and $R=\sqrt{D}$ to satisfy our requirement.
\end{proof}

\section{ \label{appendix:convexity} Appendix~II: Proof of Theorem~\ref{thm:hyperplane} }

As we have shown in the main text, to prove the convexity of $\mathcal{C}$, it suffices to prove Theorem~\ref{thm:hyperplane}, i.e., 
for any full rank state $\sigma_{AB}\in \partial\mathcal{C}$ and any state $\rho_{AB}\in \mathcal{C}$,
\begin{eqnarray*}
\tr \big((\sqrt{\det\sigma_{AB}} \sigma_{AB}^{-1}-\sigma_{AB}+\sigma_B)\rho_{AB}\big)\geq 0.
\end{eqnarray*}

To prove theorem~\ref{thm:hyperplane}, our main strategy is as follows: we first restate Theorem~\ref{thm:hyperplane} as the non-negativity of a multivariable function on some specified region and then apply a step by step optimization procedure to the objective function. In each step, we fix several variables and think of objective function as a one-variable function whose minimum point can be easily computed. Thus one variable will be eliminated within each step. By repeating this procedure several times, we could greatly simplify the objetive function as well as the constraints.

\begin{proof}

As we have seen in Appendix.~I, we can parameterize points in $\mathcal{C}$ by using three $2$-by-$2$ matrices.

Thus, we can write
\begin{eqnarray}
\rho_{AB}=\left(
\begin{array}{cc}
Q_1  &  R_1   \\
P_1  &  0
\end{array}
\right)\left(
\begin{array}{cc}
Q_1^{\dagger}  &  P_1   \\
R_1  &  0
\end{array}
\right)
\end{eqnarray}
and
\begin{eqnarray}
\sigma_{AB}=\left(
\begin{array}{cc}
Q_2  &  R_2   \\
P_2  &  0
\end{array}
\right)\left(
\begin{array}{cc}
Q_2^{\dagger}  &  P_2   \\
R_2  &  0
\end{array}
\right)
\end{eqnarray}
where $P_1,Q_1,R_1,P_2,Q_2,R_2\in \mathbb{M}_2$ satisfies $\Vert P_1R_1\Vert_{\mathrm{tr}}^2 \geq \Vert P_1Q_1^{\dagger}\Vert_{\mathrm{tr}}^2-\Vert P_1Q_1\Vert_{\mathrm{tr}}^2$, $\Vert P_2R_2\Vert_{\mathrm{tr}}^2= \Vert P_2Q_2^{\dagger}\Vert_{\mathrm{tr}}^2-\Vert P_2Q_2\Vert_{\mathrm{tr}}^2$ and $P_1,R_1,P_2,R_2\geq 0$.

Under our assumption, $\sigma_{AB}$ has full rank, thus
\begin{eqnarray*}
\sigma_{AB}^{-1}=\left(
\begin{array}{cc}
0  &  R_2^{-1}   \\
P_2^{-1}  &  -P_2^{-1}Q_2^{\dagger}R_2^{-1}
\end{array}
\right)\left(
\begin{array}{cc}
0  &  P_2^{-1}   \\
R_2^{-1}  &  -R_2^{-1}Q_2P_2^{-1}
\end{array}
\right).
\end{eqnarray*}

Hence, $\tr\Big((\sqrt{\det\sigma_{AB}} \sigma_{AB}^{-1}- \sigma_{AB}+\sigma_B)\rho_{AB}\Big)$ can be written as
\begin{eqnarray}
&&\tr(A\cdot(Q_1Q_1^{\dagger}+R_1^2))-\tr(B\cdot Q_1P_1)-\tr(P_1Q_1^{\dagger}\cdot B^{\dagger})\nonumber\\
&&+\tr(C\cdot P_1^2)\label{eq:block}
\end{eqnarray}
where
\begin{eqnarray}
A&=&\left(\begin{array}{cc}a_{11} & a_{12} \\a_{21} & a_{22}\end{array}\right)=\det(P_2R_2)R_2^{-2}+P_2^2;\nonumber\\
B&=&\left(\begin{array}{cc}b_{11} & b_{12} \\b_{21} & b_{22}\end{array}\right)=\det(P_2R_2)P_2^{-1}Q_2^{\dagger}R_2^{-2}+P_2Q_2^{\dagger};\nonumber\\
C&=&\left(\begin{array}{cc}c_{11} & c_{12} \\c_{21} & c_{22}\end{array}\right)=Q_2Q_2^{\dagger}+R_2^2\nonumber\\
&&+\det(P_2R_2)(P_2^{-2}+P_2^{-1}Q_2^{\dagger}R_2^{-2}Q_2P_2^{-1}).\label{eq:ABC}
\end{eqnarray}

We will denote our objective function Eq.~\eqref{eq:block} as $\tau(P_1,Q_1,R_1,P_2,Q_2,R_2)$. We will prove $\tau(P_1,Q_1,R_1,P_2,Q_2,R_2)\geq 0$ under the assumption that $\Vert P_1R_1\Vert_{\mathrm{tr}}^2 \geq \Vert P_1Q_1^{\dagger}\Vert_{\mathrm{tr}}^2-\Vert P_1Q_1\Vert_{\mathrm{tr}}^2$, $\Vert P_2R_2\Vert_{\mathrm{tr}}^2= \Vert P_2Q_2^{\dagger}\Vert_{\mathrm{tr}}^2-\Vert P_2Q_2\Vert_{\mathrm{tr}}^2$ and $P_1,R_1,P_2,R_2\geq 0$.

To prove the desired conditional inequality, let us first fix $P_1,Q_1,P_2,Q_2,R_2$ and minimize $\tau(P_1,Q_1,R_1,P_2,Q_2,R_2)$ subject to $\Vert P_1R_1\Vert_{\mathrm{tr}}^2 \geq \Vert P_1Q_1^{\dagger}\Vert_{\mathrm{tr}}^2-\Vert P_1Q_1\Vert_{\mathrm{tr}}^2$. In this step, we only need to consider the terms involving $R_1$, i.e., we will minimize $\tr(A \cdot R_1^2)$ subject to $\Vert P_1R_1\Vert_{\mathrm{tr}}^2 \geq \Vert P_1Q_1^{\dagger}\Vert_{\mathrm{tr}}^2-\Vert P_1Q_1\Vert_{\mathrm{tr}}^2$.

If $\Vert P_1Q_1^{\dagger}\Vert_{\mathrm{tr}}\leq \Vert P_1Q_1\Vert_{\mathrm{tr}}$, there is no constraint on $R_1$. Trivially, we have $\tr(A \cdot R_1^2)\geq 0$.

Now let us investigate the non-trivial situation that $\Vert P_1Q_1^{\dagger}\Vert_{\mathrm{tr}}>\Vert P_1Q_1\Vert_{\mathrm{tr}}$.

Let $\mathbb{U}_2$ denote the set of $2$-by-$2$ unitary matrices. According to the Cauchy-Schwarz inequality, we have
\begin{eqnarray*}
&&\tr(A \cdot R_1^2) \cdot \tr(A^{-1}P_1^2)\\
&=&\max\limits_{U,V\in \mathbb{U}_2} (\tr(U^{\dagger}R_1 A R_1U)\cdot \tr(V^{\dagger}P_1A^{-1}P_1V))\\
&\geq& \max\limits_{U,V\in \mathbb{U}_2}\vert\tr(V^{\dagger}P_1R_1 U)\vert^2\\
&=& \Vert P_1R_1\Vert_{\mathrm{tr}}^2\geq \tr(P_1^2(Q_1^{\dagger}Q_1-Q_1Q_1^{\dagger})).
\end{eqnarray*}

This implies
\begin{eqnarray*}
&&\tr(A R_1^2) \geq\frac{\tr(P_1^2(Q_1^{\dagger}Q_1-Q_1Q_1^{\dagger}))}{ \tr(A^{-1}P_1^2)}\\
\end{eqnarray*}
and the equality holds only if there exist $U,V \in \mathbb{U}_2$ such that $A^{\frac{1}{2}}R_1U$ and $A^{-\frac{1}{2}}P_1V$ are linearly dependent, $V^{\dagger}P_1R_1 U$ is diagonal and $\Vert P_1R_1\Vert_{\mathrm{tr}}^2 =\Vert P_1Q_1^{\dagger}\Vert_{\mathrm{tr}}^2-\Vert P_1Q_1\Vert_{\mathrm{tr}}^2$.

Thus by combining the two situations together, we have
\begin{eqnarray}
\tr(A R_1^2) &\geq& \max\Big\{0,\frac{\tr(P_1^2(Q_1^{\dagger}Q_1-Q_1Q_1^{\dagger}))}{ \tr(A^{-1}P_1^2)}\Big\}\\
&\geq& \frac{\tr(P_1^2(Q_1^{\dagger}Q_1-Q_1Q_1^{\dagger}))}{ \tr(A^{-1}P_1^2)}.\label{eq:opt1}
\end{eqnarray}

As a consequence, it suffices to prove
\begin{eqnarray}
&&\tr\Big((Q_1^{\dagger}A^{\frac{1}{2}}-P_1BA^{-\frac{1}{2}})(A^{\frac{1}{2}}Q_1-A^{-\frac{1}{2}}B^{\dagger}P_1)\Big)\nonumber\\
&+&\tr((C-BA^{-1}B^{\dagger})P_1^2)+\frac{\tr(P_1^2(Q_1^{\dagger}Q_1-Q_1Q_1^{\dagger}))}{ \tr(A^{-1}P_1^2)}\geq 0 \nonumber\\
\label{eq:ineq3}
\end{eqnarray}
for any $P_1,Q_1\in \mathbb{M}_2$ and $P_1\geq 0$.

Without loss of generality, we can always assume $P_1$ is diagonal. Let $P_1=\left(\begin{array}{cc}x & 0\\
0 & y\end{array}\right)$ and $Q_1=\left(\begin{array}{cc}q_{11}&q_{12}\\q_{21}& q_{22}\end{array}\right)$. Note that, $q_{11}$ and $q_{22}$ only appear in the first term, i.e., $\tr\Big((Q_1^{\dagger}A^{\frac{1}{2}}-P_1BA^{-\frac{1}{2}})(A^{\frac{1}{2}}Q_1-A^{-\frac{1}{2}}B^{\dagger}P_1)\Big)$. We thus choose suitable $q_{11}$ and $q_{22}$ to minimize $\tr\Big((Q_1^{\dagger}A^{\frac{1}{2}}-P_1BA^{-\frac{1}{2}})(A^{\frac{1}{2}}Q_1-A^{-\frac{1}{2}}B^{\dagger}P_1)\Big)$.

Here we divide $Q_1$ into the diagonal part $\widehat{Q}_1=\left(\begin{array}{cc} q_{11}&0\\0 & q_{22}
\end{array}\right)$ and antidiagonal part $\widetilde{Q}_1=\left(\begin{array}{cc}
0 & q_{12} \\q_{21} & 0\end{array}\right)$, then
\begin{eqnarray*}
&&\Vert A^{\frac{1}{2}}Q_1-A^{-\frac{1}{2}}B^{\dagger}P_1\Vert_F\\
&=&\Big\Vert q_{11}A^{\frac{1}{2}}\ket{0}\bra{0}+q_{22}A^{\frac{1}{2}}\ket{1}\bra{1}+A^{\frac{1}{2}}\widetilde{Q}_1-A^{-\frac{1}{2}}B^{\dagger}P_1\Big\Vert_F
\end{eqnarray*}
which can be considered as the distance from a point $\big(-A^{\frac{1}{2}}\widetilde{Q}_1+A^{-\frac{1}{2}}B^{\dagger}P_1\big)$ to another point on the plane spanned by $A^{\frac{1}{2}}\ket{0}\bra{0}$ and $A^{\frac{1}{2}}\ket{1}\bra{1}$.

Certainly, the minimum can be achieved if and only if $q_{11}A^{\frac{1}{2}}\ket{0}\bra{0}+q_{22}A^{\frac{1}{2}}\ket{1}\bra{1}$ is the projection of $\big(-A^{\frac{1}{2}}\widetilde{Q}_1+A^{-\frac{1}{2}}B^{\dagger}P_1\big)$ onto the plane, i.e., $q_{11}A^{\frac{1}{2}}\ket{0}\bra{0}+q_{22}A^{\frac{1}{2}}\ket{1}\bra{1}+A^{\frac{1}{2}}\widetilde{Q}_1-A^{-\frac{1}{2}}B^{\dagger}P_1 \perp \mathop{\mathrm{span}}\{A^{\frac{1}{2}}\ket{0}\bra{0}, A^{\frac{1}{2}}\ket{1}\bra{1}\}$.

Thus by solving the linear system derived by the orthogonal conditions, we have

\begin{widetext}
\begin{eqnarray}
&&\min\limits_{q_{11},q_{22}}\Vert q_{11}A^{\frac{1}{2}}\ket{0}\bra{0}+q_{22}A^{\frac{1}{2}}\ket{1}\bra{1}+A^{\frac{1}{2}}\widetilde{Q}_1-A^{-\frac{1}{2}}B^{\dagger}P_1\Vert_F^2\nonumber\\
&=&\Big\Vert A^{\frac{1}{2}}\widetilde{Q}_1-A^{-\frac{1}{2}}B^{\dagger}P_1\Big\Vert_F^2-\Big\Vert (-\frac{\bra{0}A\widetilde{Q}_1-B^{\dagger}P_1\ket{0}}{\bra{0}A\ket{0}})A^{\frac{1}{2}}\ket{0}\bra{0}
+(-\frac{\bra{1}A\widetilde{Q}_1-B^{\dagger}P_1\ket{1}}{\bra{1}A\ket{1}})A^{\frac{1}{2}}\ket{1}\bra{1}\Big\Vert_F^2.\label{eq:opt2}
\end{eqnarray}
\end{widetext}

By substituting corresponding terms in the left-hand side of Eq.~\eqref{eq:ineq3}, we have
\begin{widetext}
\begin{eqnarray}
&&\tau(P_1,Q_1,R_1,P_2,Q_2,R_2)\nonumber\\
&\geq&\tr\Big((Q_1^{\dagger}A^{\frac{1}{2}}-P_1BA^{-\frac{1}{2}})(A^{\frac{1}{2}}Q_1-A^{-\frac{1}{2}}B^{\dagger}P_1)\Big)
+\tr((C-BA^{-1}B^{\dagger})P_1^2)+\frac{\tr(P_1^2(Q_1^{\dagger}Q_1-Q_1Q_1^{\dagger}))}{ \tr(A^{-1}P_1^2)}\nonumber\\
&\geq &\Big\Vert A^{\frac{1}{2}}\widetilde{Q}_1-A^{-\frac{1}{2}}B^{\dagger}P_1\Big\Vert_F^2-\Big\Vert (-\frac{\bra{0}A\widetilde{Q}_1-B^{\dagger}P_1\ket{0}}{\bra{0}A\ket{0}})A^{\frac{1}{2}}\ket{0}\bra{0}
+(-\frac{\bra{1}A\widetilde{Q}_1-B^{\dagger}P_1\ket{1}}{\bra{1}A\ket{1}})A^{\frac{1}{2}}\ket{1}\bra{1}\Big\Vert_F^2\nonumber\\
&&+\tr((C-BA^{-1}B^{\dagger})P_1^2)+\frac{\tr(P_1^2(Q_1^{\dagger}Q_1-Q_1Q_1^{\dagger}))}{ \tr(A^{-1}P_1^2)}\nonumber\\
&=& \tr(A\widetilde{Q}_1\widetilde{Q}_1^{\dagger})-\tr(B^{\dagger}P_1\widetilde{Q}_1^{\dagger})-\tr(B\widetilde{Q}_1P_1)
-\frac{\vert\bra{0}A\widetilde{Q}_1-B^{\dagger}P_1\ket{0}\vert^2}{\bra{0}A\ket{0}}-\frac{\vert\bra{1}A\widetilde{Q}_1-B^{\dagger}P_1\ket{1}\vert^2}{\bra{1}A\ket{1}}\nonumber\\
&&+\tr(CP_1^2)+\frac{\tr(P_1^2(Q_1^{\dagger}Q_1-Q_1Q_1^{\dagger}))}{ \tr(A^{-1}P_1^2)}\nonumber
\end{eqnarray}
\begin{eqnarray}
&=& c_{11}x^2+c_{22}y^2+a_{11} \vert q_{12}\vert^2+a_{22}\vert q_{21}\vert^2-b_{21}q_{12}y -b_{12}q_{21}x\nonumber\\
&&-b_{21}^{\ast}q_{12}^{\ast}y-b_{12}^{\ast}q_{21}^{\ast}x+\frac{\det(A)(x^2-y^2)(\vert q_{21}\vert^2-\vert q_{12}\vert^2)}{a_{11}y^2 +a_{22} x^2 }
-\frac{\vert q_{21}a_{12}-x b_{11}^{\ast}\vert^2}{a_{11}}-\frac{\vert q_{12}a_{21}-y b_{22}^{\ast}\vert^2}{a_{22}}\nonumber\\
&=&\frac{\det(A)(a_{11}+a_{22})}{a_{11}y^2+a_{22}x^2}\nonumber\\
&&\cdot\Bigg(\frac{1}{a_{11}}\Big\vert x q_{21} +\frac{(a_{11}y^2+a_{22}x^2) (a_{12}b_{11}-b_{12} a_{11})^{\ast}}{\det(A)(a_{11}+a_{22})}\Big\vert^2
+\frac{1}{a_{22}}\Big\vert y q_{12} +\frac{(a_{11}y^2+a_{22}x^2) (a_{21}b_{22}-a_{22}b_{21})^{\ast}}{\det(A)(a_{11}+a_{22})}\Big\vert^2\Bigg)\nonumber\\
&&+(c_{11}-\frac{\vert b_{11}\vert^2}{a_{11}}-\frac{\frac{a_{22}}{a_{11}} \vert a_{12}b_{11}-a_{11}b_{12}\vert^2+\vert a_{21}b_{22}-a_{22}b_{21}\vert^2}{\det(A)(a_{11}+a_{22})})
 x^2\nonumber\\
&&+(c_{22}-\frac{\vert b_{22}\vert^2}{a_{22}}-\frac{\vert a_{12}b_{11}-a_{11}b_{12}\vert^2+\frac{a_{11}}{a_{22}} \vert a_{21}b_{22}-a_{22}b_{21}\vert^2}{\det(A)(a_{11}+a_{22})})
 y^2. \label{eq:opt3}
\end{eqnarray}
\end{widetext}

To complete our proof, we will show the last two terms all vanish when the full rank state satisfies $\sigma_{AB}\in \partial{\mathcal{C}}$, which will immediately lead to our desired conditional inequality.

Note that $\left(\begin{array}{cc}A & -B^{\dagger} \\ -B & C\end{array}\right)$ represents the matrix form of $H_{AB}=\sqrt{\det{\sigma_{AB}}}\sigma_{AB}^{-1}-\sigma_{AB}+\sigma_B$. Thus the last two terms vanish if and only if
\begin{eqnarray*}
&&\det \bra{0_B} H_{AB} \ket{0_B}=\det \bra{1_B} H_{AB} \ket{1_B}\\
&=& \frac{a_{22}\Big\vert \det{\bra{0_B}H_{AB}\ket{0_A}}\Big\vert^2+a_{11}\Big\vert \det{\bra{1_B}H_{AB}\ket{0_A}}\Big\vert^2}{\det{\bra{0_A}H_{AB}\ket{0_A}}(a_{11}+a_{22})}.
\end{eqnarray*}

Let $H_{AB}^{(i_1,\cdots,i_k)}$ be the submatrix formed by taking the $(i_1,\cdots,i_k)$-th rows and columns of $H_{AB}$. Then $\det {\bra{0_B} H_{AB} \ket{0_B}}=\det {\bra{1_B} H_{AB} \ket{1_B}}$ means $\det{H_{AB}^{(1,3)}}=\det{H_{AB}^{(2,4)}}$. Once we have proved the first equality, the second equality can be rewritten as
\begin{eqnarray*}
&&a_{22}\det{\bra{0_A}H_{AB}\ket{0_A}}\det {\bra{0_B} H_{AB} \ket{0_B}}\\
&&+a_{11}\det{\bra{0_A}H_{AB}\ket{0_A}}\det {\bra{1_B} H_{AB} \ket{1_B}}\\
&=&a_{22}\Big\vert \det{\bra{0_B}H_{AB}\ket{0_A}}\Big\vert^2+a_{11}\Big\vert \det{\bra{1_B}H_{AB}\ket{0_A}}\Big\vert^2
\end{eqnarray*}
which can be further reformulated as $\det{H_{AB}^{(1,2,3)}}=-\det{H_{AB}^{(1,2,4)}}$.

Thus, to accomplish our goal, it suffices to prove
\begin{eqnarray}
\det{H_{AB}^{(1,3)}}&=&\det{H_{AB}^{(2,4)}};\label{eq:term1}\\
\det{H_{AB}^{(1,2,3)}}&=&-\det{H_{AB}^{(1,2,4)}}.\label{eq:term2}
\end{eqnarray}

For Eq.~\eqref{eq:term1}, i.e.
\begin{eqnarray*}
\bra{0}A\ket{0} \bra{0}C\ket{0}-\vert \bra{0} B\ket{0}\vert^2=\bra{1}A\ket{1} \bra{1}C\ket{1}-\vert \bra{1}B\ket{1}\vert^2,
\end{eqnarray*}
it is equivalent to
\begin{eqnarray*}
\bra{0} AC\ket{0}+\bra{1} CA\ket{1}=\bra{1} B^{\dagger}B\ket{1}+\bra{0} BB^{\dagger}\ket{0}.
\end{eqnarray*}

To prove this, it suffices to show $AC-BB^{\dagger}$ is the adjugate matrix of $CA-B^{\dagger}B$, i.e. $AC-BB^{\dagger}+CA-B^{\dagger}B=\tr(AC-BB^{\dagger})\mathbb{I}$.

In fact, to prove the above claim, our assumption $\Vert P_2R_2\Vert_{\mathrm{tr}}^2= \Vert P_2Q_2^{\dagger}\Vert_{\mathrm{tr}}^2-\Vert P_2Q_2\Vert_{\mathrm{tr}}^2$ is not necessary. The identity holds for any $2$-by-$2$ Hermitian matrices $P_2$, $R_2$ and any $2$-by-$2$ matrix $Q_2$. This fact can be easily verified by using symbolic computing softwares like Mathematica\footnote{The Mathematica notebook can be found at \url{http://jianxin.iqubit.org/downloads/Verification.nb}.}.

Now let us look at Eq.~\eqref{eq:term2}. Let
\begin{eqnarray*}
\widetilde{H}&=&\left(\begin{array}{cc} A & 0 \\0 & C-BA^{-1}B^{\dagger}\end{array}\right)\\
&=&\left(\begin{array}{cc}
\mathbb{I} & 0\\
BA^{-1} & \mathbb{I}
\end{array}\right) H \left(\begin{array}{cc}
\mathbb{I} & A^{-1}B^{\dagger}\\
0 & \mathbb{I}
\end{array}\right).
\end{eqnarray*}

The determinant is invariant under elementary row and column operations, we have
$\det H_{AB}^{(1,2,3)}=\det \widetilde{H}^{(1,2,3)}=\det{(A)} \widetilde{H}_{3,3}$ and $\det H_{AB}^{(1,2,4)}=\det \widetilde{H}^{(1,2,4)}=\det{(A)} \widetilde{H}_{4,4}$. Therefore, Eq.~\eqref{eq:term2} is equivalent to $\widetilde{H}_{3,3}=-\widetilde{H}_{4,4}$, i.e.
\begin{eqnarray*}\label{eq:term2tr}
\tr(C-BA^{-1}B^{\dagger})=0.
\end{eqnarray*}

$\tr(C-BA^{-1}B^{\dagger})$ is invariant under local unitary operations, thus it suffices to prove $\tr(C-BA^{-1}B^{\dagger})=0$ for diagonal $P_2$.

Again, let $P_2=\left(\begin{array}{cc}x^{\prime} & 0\\ 0 & y^{\prime}\end{array}\right)$ and divide $Q_2=\left(\begin{array}{cc}q_{11}^{\prime}&q_{12}^{\prime}\\q_{21}^{\prime}& q_{22}^{\prime}\end{array}\right)$ into the diagonal part $\widehat{Q}_2$ and antidiagonal part $\widetilde{Q}_2$. Simple calculation will show that $\widehat{Q}_2$ all cancel out in $\tr(C-BA^{-1}B^{\dagger})$ so we can assume $q_{11}^{\prime}=q_{22}^{\prime}=0$ without loss of generality.

Then everything is straightforward.

By substituting $P_2=\left(\begin{array}{cc}x^{\prime} & 0\\ 0 & y^{\prime}\end{array}\right)$, $Q_2=\left(\begin{array}{cc}0&q_{12}^{\prime}\\q_{21}^{\prime}&0\end{array}\right)$ and $R_2=\left(\begin{array}{cc}r_{11} & r_{12} \\ r_{12}^{\ast} & r_{22}\end{array}\right)$ in Eq.~\eqref{eq:ABC}, we will have
\begin{eqnarray*}
&&\tr(C-BA^{-1}B^{\dagger})\\
&=&\frac{(r_{11}x^{\prime}+r_{22}y^{\prime})(r_{11}y^{\prime}+r_{22}x^{\prime})-\vert r_{12}\vert^2(x^{\prime}-y^{\prime})^2}{x^{\prime}y^{\prime}\Big((r_{11}x^{\prime}+r_{22} y^{\prime})^2+\vert r_{12}\vert^2 (x^{\prime}-y^{\prime})^2\Big)}\\
&&\cdot \Big( (r_{11}x^{\prime}+r_{22}y^{\prime})^2+\vert r_{12}\vert^2(x^{\prime}-y^{\prime})^2\\
&&-((x^{\prime})^2-(y^{\prime})^2)(\vert q_{21}^{\prime}\vert^2-\vert q_{12}^{\prime}\vert^2)\Big).
\end{eqnarray*}

Under our assumption, a full rank state $\sigma_{AB}\in \partial{\mathcal{C}}$ implies $\Vert P_2R_2\Vert_{\mathrm{tr}}^2= \Vert P_2Q_2^{\dagger}\Vert_{\mathrm{tr}}^2-\Vert P_2Q_2\Vert_{\mathrm{tr}}^2$, or equivalently $(r_{11}x^{\prime}+r_{22}y^{\prime})^2+\vert r_{12}\vert^2(x^{\prime}-y^{\prime})^2=((x^{\prime})^2-(y^{\prime})^2)(\vert q_{21}^{\prime}\vert^2-\vert q_{12}^{\prime}\vert^2)$. $\tr(C-BA^{-1}B^{\dagger})=0$ follows immediately.
\end{proof}

\section{Appendix~III: Faces of $\mathcal{C}$}
Follows from Theorem~\ref{thm:hyperplane}, $\mathcal{C}$ is a convex body. Faces of $\mathcal{C}$ are its intersections with the supporting hyperplanes.

Let us start with a full rank boundary point $\sigma_{AB} \in \partial\mathcal{C}$. Let $H_{AB}(\sigma_{AB})=\sqrt{\det{\sigma_{AB}}}\sigma_{AB}^{-1}-\sigma_{AB}+\sigma_B$,  then the supporting hyperplane
\begin{eqnarray*}
\mathcal{L}(\sigma_{AB}):=\{X: \tr (H_{AB}(\sigma_{AB})\cdot X)=0\}
\end{eqnarray*}
also defines a face $\mathcal{F}(\sigma_{AB})=\mathcal{L}(\sigma_{AB})\bigcap \mathcal{C}$.

Recall that in Appendix.~II, we applied a step-by-step optimization procedure to prove $\tr(H_{AB}(\sigma_{AB}) \cdot \rho_{AB})\geq 0$ for any $\rho_{AB}\in \mathcal{C}$. Thus, $\mathcal{L}(\sigma_{AB})\bigcap \mathcal{C}$ contains all those states satisfying equality in every optimization step. In this Appendix, we will solve the equation system and then
provide a complete parameterization of $\mathcal{F}(\sigma_{AB})$. As a byproduct, we will prove Theorem~\ref{thm:suf} at the end of this Appendix.

According to Appendix.~I, $\sigma_{AB}$ can be represented as the following by using three $2$-by-$2$ matrices $P_2, Q_2, R_2$ satisfying $\Vert P_2R_2\Vert_{\mathrm{tr}}^2= \Vert P_2Q_2^{\dagger}\Vert_{\mathrm{tr}}^2-\Vert P_2Q_2\Vert_{\mathrm{tr}}^2$ and $P_2, R_2\geq 0$:
\begin{eqnarray*}
\sigma_{AB}=\left(
\begin{array}{cc}
Q_2  &  R_2   \\
P_2  &  0
\end{array}
\right)\left(
\begin{array}{cc}
Q_2^{\dagger}  &  P_2   \\
R_2  &  0
\end{array}
\right)\in \partial{\mathcal{C}}.
\end{eqnarray*}

We can represent any state $\rho_{AB}\in \mathcal{F}(\sigma_{AB})$ in the same way:
\begin{eqnarray*}
\rho_{AB}=\left(
\begin{array}{cc}
Q_1  &  R_1   \\
P_1  &  0
\end{array}
\right)\left(
\begin{array}{cc}
Q_1^{\dagger}  &  P_1\\
R_1  &  0
\end{array}
\right).
\end{eqnarray*}

Thus, our aim is to characterize the set of $3$-tuples $\{(P_1,Q_1,R_1): \left(
\begin{array}{cc}
Q_1  &  R_1   \\
P_1  &  0
\end{array}
\right)\left(
\begin{array}{cc}
Q_1^{\dagger}  &  P_1   \\
R_1  &  0
\end{array}
\right) \in \mathcal{F}(\sigma_{AB})\}$ for any given $\sigma_{AB}=\left(
\begin{array}{cc}
Q_2  &  R_2   \\
P_2  &  0
\end{array}
\right)\left(
\begin{array}{cc}
Q_2^{\dagger}  &  P_2   \\
R_2  &  0
\end{array}
\right)\in \partial{\mathcal{C}}$, or equivalently, those $3$-tuples $(P_1,Q_1,R_1)$ to make $\tau(P_1,Q_1,R_1,P_2,Q_2,R_2)$ which is defined in Eq.~\eqref{eq:block} vanish.


We first consider those $3$-tuples $(P_1,Q_1,R_1)$ in which $P_1$ is a diagonal matrix $\left(\begin{array}{cc}x & 0\\0 & y\end{array}\right)$. It is also what we assumed in our proof in Appendix.~II. $A=(a_{ij})_{1\leq i,j\leq 2}$ and $B=(b_{ij})_{1\leq i,j\leq 2}$ are matrices only depending on $P_2,Q_2,R_2$, as given in Eq~\eqref{eq:ABC}. As we provide a step-by-step optimization procedure to show $\tau(P_1,Q_1,R_1,P_2,Q_2,R_2)\geq 0$ in Appendix~II, $Q_1$ and $R_1$ must be chosen to make the equalities hold in every optimization step. 
\begin{enumerate}
\item [1.] The equality in Eq~\eqref{eq:opt1} holds if and only if there exist $U,V \in \mathbb{U}_2$ such that $A^{\frac{1}{2}}R_1U$ and $A^{-\frac{1}{2}}P_1V$ are linearly dependent, $V^{\dagger}P_1R_1 U$ is diagonal and $\Vert P_1R_1\Vert_{\mathrm{tr}}^2 =\Vert P_1Q_1^{\dagger}\Vert_{\mathrm{tr}}^2-\Vert P_1Q_1\Vert_{\mathrm{tr}}^2$;
\item [2.] The minimum of the left-hand-side in Eq~\eqref{eq:opt2} can be achieved if and only if $q_{11}A^{\frac{1}{2}}\ket{0}\bra{0}+q_{22}A^{\frac{1}{2}}\ket{1}\bra{1}$ is the projection of $\big(-A^{\frac{1}{2}}\widetilde{Q}_1+A^{-\frac{1}{2}}B^{\dagger}P_1\big)$ onto the plane, i.e., $q_{11}A^{\frac{1}{2}}\ket{0}\bra{0}+q_{22}A^{\frac{1}{2}}\ket{1}\bra{1}+A^{\frac{1}{2}}\widetilde{Q}_1-A^{-\frac{1}{2}}B^{\dagger}P_1 \perp \mathop{\mathrm{span}}\{A^{\frac{1}{2}}\ket{0}\bra{0}, A^{\frac{1}{2}}\ket{1}\bra{1}\}$;
\item [3.] The right-hand-side of Eq~\eqref{eq:opt3} equals to zero if and only if $x q_{21} +\frac{(a_{11}y^2+a_{22}x^2) (a_{12}b_{11}-b_{12} a_{11})^{\ast}}{\det(A)(a_{11}+a_{22})}$ and $y q_{12} +\frac{(a_{11}y^2+a_{22}x^2) (a_{21}b_{22}-a_{22}b_{21})^{\ast}}{\det(A)(a_{11}+a_{22})}$ all vanish.
\end{enumerate}

$Q_1$ and $R_1$ can thus be derived by using elementary linear algebra. Explicit expressions will be given later in the more general Lemma~\ref{lemma:pa}.

If $P_1$ is not diagonal, then follows from the eigenvalue decomposition, we can write $P_1=U\left(
\begin{array}{cc}
x  &  0   \\
0  &  y
\end{array}
\right)U^{\dagger}$ where $U$ is a $2$-by-$2$ unitary matrix and $x, y$ are positive numbers. Note that $\rho_{AB}\in \mathcal{F}(\sigma_{AB})$ if and only if $(U^{\dagger}\otimes U^{\dagger})\rho_{AB}(U\otimes U)\in \mathcal{F}\big((U^{\dagger}\otimes U^{\dagger})\sigma_{AB}(U\otimes U)\big)$ and $(U^{\dagger}\otimes U^{\dagger})\rho_{AB}(U\otimes U)$ can be represented by the $3$-tuple $(U^{\dagger}P_1U, U^{\dagger}Q_1U, U^{\dagger}R_1U)$, hence our result for diagonal case will apply directly.

To summarize, given a full rank $\sigma_{AB}=\left(
\begin{array}{cc}
Q_2  &  R_2   \\
P_2  &  0
\end{array}
\right)\left(
\begin{array}{cc}
Q_2^{\dagger}  &  P_2   \\
R_2  &  0
\end{array}
\right)\in \partial{\mathcal{C}}$, we can parameterize all full rank states in $\mathcal{F}(\sigma_{AB})$ by using a $2$-by-$2$ unitary matrix $U$ and positive numbers $x, y$ as the following lemma:
\begin{lemma}\label{lemma:pa}
All full rank states in $\mathcal{F}(\sigma_{AB})$ can be represented as some
\begin{eqnarray*}
\widetilde{\rho}_{AB}(x,y,U)=\left(
\begin{array}{cc}
Q_1  &  R_1   \\
P_1  &  0
\end{array}
\right)\left(
\begin{array}{cc}
Q_1^{\dagger}  &  P_1\\
R_1  &  0
\end{array}
\right)
\end{eqnarray*}
where 
\begin{eqnarray*}
P_1&=&U\left(
\begin{array}{cc}
x  &  0   \\
0  &  y
\end{array}
\right)U^{\dagger};\\
Q_1&=&\frac{1}{\det(A)\tr(A)}U\left(\begin{array}{cc}
q_{11} & q_{12}\\
q_{21} & q_{22}
\end{array}
\right)\cdot \left(
\begin{array}{cc}
x & 0\\
0 & y\\
\end{array}
\right)^{-1}U^{\dagger};\\
\end{eqnarray*}
and
\begin{widetext}
\begin{eqnarray*}
R_1=\frac{\sqrt{(x^2-y^2)(\vert \frac{a_{12}b_{11}-b_{12} a_{11}}{x} \vert^2-\vert  \frac{a_{21}b_{22}-a_{22}b_{21}}{y}\vert^2)}}{\det(A)\tr(A)}
\cdot U\sqrt{\left(
\begin{array}{cc}
a_{22}^2x^2+\vert a_{12}\vert^2 y^2 &  -a_{12}(a_{11}y^2+a_{22}x^2)   \\
-a_{21}(a_{11}y^2+a_{22}x^2)  & \vert a_{21}\vert^2 x^2+a_{11}^2 y^2
\end{array}
\right)}U^{\dagger}
\end{eqnarray*}
\end{widetext}
where
\begin{eqnarray*}
A(U)&=&(a_{ij})_{1\leq i,j\leq 2}=U^{\dagger}(\det(P_2R_2)R_2^{-2}+P_2^2)U;\\
B(U)&=&(b_{ij})_{1\leq i,j\leq 2}=U^{\dagger}(\det(P_2R_2)P_2^{-1}Q_2^{\dagger}R_2^{-2}+P_2Q_2^{\dagger})U
\end{eqnarray*}
and 
\begin{eqnarray*}
q_{11}&=&((a_{11}a_{22}+a_{22}^2-a_{12}a_{21})b_{11}^{\ast}-a_{12}a_{22}b_{12}^{\ast})x^2\\
&&+a_{12}(a_{21}b_{11}^{\ast}-a_{11}b_{12}^{\ast})y^2;\\
q_{12}&=&-(a_{11}y^2+a_{22}x^2) (a_{21}b_{22}-a_{22}b_{21})^{\ast};\\
q_{21}&=& -(a_{11}y^2+a_{22}x^2) (a_{12}b_{11}-b_{12} a_{11})^{\ast};\\
q_{22}&=&a_{21}(a_{12}b_{22}^{\ast}-a_{22}b_{21}^{\ast})x^2\\
&&+((a_{11}a_{22}+a_{11}^2-a_{12}a_{21})b_{22}^{\ast}-a_{21}a_{11}b_{21}^{\ast})y^2.
\end{eqnarray*}
\end{lemma}
We reuse the symbols `$a_{ij}$' and `$b_{ij}$' to keep our formulas simple, but one should keep in mind that they depend on unitary matrix $U$. Indeed, we should use the more precise form $a_{ij}(U)$ and $b_{ij}(U)$ instead in Lemma~\ref{lemma:pa} if we do not care about the length of the expressions.

To make sure that $\widetilde{\rho}_{AB}(x,y,U)$ lies in $\mathcal{C}$,  $x$ and $y$ must satisfy
\begin{eqnarray*}
(x-y)(\vert a_{21}b_{22}-a_{22}b_{21}\vert x-\vert a_{12}b_{11}-b_{12} a_{11}\vert y) \leq 0.
\end{eqnarray*}

All full rank states in $\mathcal{F}(\sigma_{AB})$ can be parameterized in this way. However, for the case $x=y$ or $\frac{x}{y}=\vert\frac{ a_{12}b_{11}-b_{12} a_{11}}{ a_{21}b_{22}-a_{22}b_{21}}\vert$, $\widetilde{\rho}_{AB}(x,y,U)$ has rank $2$ since the corresponding $R_1$ is a zero matrix for both cases.

$\mathcal{F}(\sigma_{AB})$ also contains other non-full-rank states which corresponds to $x=0$ or $y=0$.

$y=0$ occurs only if $\vert a_{21}b_{22}-a_{22}b_{21}\vert=0$. In this case, we have
\begin{widetext}
\begin{eqnarray*}
\widetilde{\rho}_{AB}(x,0,U)
=(U\otimes U)\left(
\begin{array}{cccc}
\frac{\vert b_{11}\vert^2 x^2}{a_{12}a_{21}}& -\frac{\vert b_{11}\vert^2 x^2}{a_{22}a_{21}} &  0 &  0 \\
-\frac{\vert b_{11}\vert^2 x^2}{a_{12}a_{22}}&  \frac{\vert b_{11}\vert^2 x^2}{a_{12}a_{21}}+\frac{\vert b_{11}\vert^2 x^2}{a_{22}^2}& \frac{b_{11}^{\ast}x^2}{a_{12}}  &  0 \\
0 & \frac{b_{11}x^2}{a_{21}} &  x^2 &  0\\
0 & 0 & 0  &0
\end{array}
\right)(U^{\dagger}\otimes U^{\dagger})
\end{eqnarray*}
\end{widetext}
which is a rank $2$ state.

We have similar results for the case $x=0$.

At the end of this Appendix, we will prove Theorem~\ref{thm:suf}  as an application of our parameterization scheme. Simple calculation will show us that all entries of $\widetilde{\rho}_{AB}(x,y,U)$ are linear combinations of $x^2$ and $y^2$. Let us assume $\vert \frac{ a_{12}b_{11}-b_{12} a_{11}}{ a_{21}b_{22}-a_{22}b_{21}}\vert>1$ without loss of generality, then for any $y\leq x\leq \vert\frac{ a_{12}b_{11}-b_{12} a_{11}}{ a_{21}b_{22}-a_{22}b_{21}}\vert y$, $\rho_{AB}(x,y,U)$ is a convex combination of $\widetilde{\rho}_{AB}(y,y,U)$ and $\widetilde{\rho}_{AB}(\vert\frac{ a_{12}b_{11}-b_{12} a_{11}}{ a_{21}b_{22}-a_{22}b_{21}}\vert y,y,U)$, both of which are rank $2$ states.

In other words, after the normalization, $\widetilde{\rho}_{AB}(x,y,U)$ only depends on unitary matrix $U$ and the ratio of $x$ and $y$. Let $\rho_{AB}(1,0,U)=\widetilde{\rho}_{AB}(1,1,U)$ and $\rho_{AB}(0,1,U)=\widetilde{\rho}_{AB}(\vert a_{12}b_{11}-b_{12} a_{11}\vert,\vert a_{21}b_{22}-a_{22}b_{21}\vert,U)$, then all states on the face $\mathcal{F}(\sigma_{AB})$ can be represented as $\rho_{AB}(\lambda,1-\lambda,U)=\lambda\rho_{AB}(1,0,U)+(1-\lambda)\rho_{AB}(0,1,U)$ where $0\leq \lambda\leq 1$ and $U\in \mathbb{U}_2$.

\end{document}